\documentclass[graybox]{SNmult}

\usepackage{type1cm}        
\usepackage{makeidx}         
\usepackage{graphicx}        
\usepackage{multicol}        
\usepackage[bottom]{footmisc}

\usepackage{newtxtext}       %
\usepackage[varvw]{newtxmath}       

\makeindex             

\usepackage{url}

\usepackage{mathtools}
\usepackage{subcaption}
\usepackage[hidelinks]{hyperref}
\usepackage[nameinlink, capitalise]{cleveref}
\crefname{equation}{}{}

\newcommand{\forstermb}[1]{\mathbf{#1}}
\newcommand{\forstermr}[1]{\mathrm{#1}}
\newcommand{\forsterbs}[1]{\boldsymbol{#1}}
\newcommand{\forsterT}{{\!\top}}
\newcommand{\forsterdd}[1]{\frac{\mathrm{d}}{\mathrm{d} #1}}
\newcommand{\forsterpp}[2][]{\frac{\partial #1}{\partial #2}}

\newcommand{\forstermbh}[1]{\hat{\mathbf{#1}}}

\newcommand{\forsterA}[1][]{\forstermb{A}_\forstermr{#1}}
\newcommand{\forsterAT}[1][]{\forstermb{A}_\forstermr{#1}^{\forsterT}}
\newcommand{\forsterqC}{\forstermb{q}_\forstermr{C}}
\newcommand{\forsterphiL}{\forsterbs{\phi}_\forstermr{L}}
\newcommand{\forstergR}{\forstermb{g}_\forstermr{R}}

\newcommand{\forstervphi}{\forsterbs{\varphi}}

\newcommand{\forsteriL}{\forstermb{i}_\forstermr{L}}

\newcommand{\forsteriV}{\forstermb{i}_\forstermr{V}}
\newcommand{\forsteriVs}{\forstermb{i}_{\forstermr{V_s}}}
\newcommand{\forsteriVc}{\forstermb{i}_{\forstermr{V_c}}}

\newcommand{\forstervs}{\forstermb{v}_\forstermr{s}}
\newcommand{\forstervc}{\forstermb{v}_\forstermr{c}}

\newcommand{\forsteris}{\forstermb{i}_\forstermr{s}}
\newcommand{\forsteric}[1][]{\forstermb{i}_{\forstermr{c}_{#1}}}

\newcommand{\forsterPT}[1][]{\forstermb{P}_\forstermr{#1}^\forsterT}
\newcommand{\forsterQT}[1][]{\forstermb{Q}_\forstermr{#1}^\forsterT}

\newcommand{\forsterW}[1][]{\forstermb{W}_\forstermr{#1}}

\begin{document}
\title*{A non-intrusive approach to index-aware learning}

\author{Peter Förster\orcidID{0000-0003-0603-6082} and\\ Idoia Cortes Garcia\orcidID{0000-0001-6557-103X} and\\ Wil Schilders\orcidID{0000-0001-5838-8579} and\\ Sebastian Schöps\orcidID{0000-0001-9150-0219}}

\institute{Peter Förster and Sebastian Schöps \at Technical University of Darmstadt, Germany, \email{peter.foerster@tu-darmstadt.de}
\and Peter Förster and Idoia Cortes Garcia and Wil Schilders \at Eindhoven University of Technology, The Netherlands}
%
%

\maketitle

\abstract{We present a non-intrusive version of the index-aware learning framework introduced in~\cite{forster2023c}. Index-aware learning itself is an approach for learning the time and parameter dependent solutions of differential-algebraic equations (DAEs), in particular those describing electric circuits. A central feature of the approach is that it ensures the learned solutions to fulfill the inherent constraints of the DAE, such as e.g.~Kirchhoff's laws in the case of electric circuits. This is achieved by leveraging a decoupling of the DAE into its differential and algebraic parts, with the non-intrusive version of the approach additionally relying on results from~\cite{forster2026a,estevez2000b,garcia2022}. We illustrate the approach using a filtered buck converter as an example and compare both the intrusive and non-intrusive versions. The code for the example is openly available at~\cite{forster2026g}.
}

\section{Introduction}
\label{sec:forster_int}
Learning the time and parameter dependent solutions of electric circuits is of interest for various applications in the contexts of design verification, optimization or uncertainty quantification. In the case of SPICE-based circuit simulators, electric circuits are usually described with modified nodal analysis (MNA)~\cite{ho1975}.
Allowing for common kinds of controlled sources, see also~\cite[Eq. (3)]{forster2026a}, the equations of MNA can be written in the following form, where the $\forstermb{A}_\star$ denote incidence matrices capturing the graph structure of the circuit, $\forsterqC$, $\forstergR$ and $\forsterphiL$ are functions modeling capacitors, resistors and inductors, the unknowns are given by the nodal potentials $\forstervphi$, the inductor currents $\forsteriL$ and the currents through (independent and controlled) voltage sources $\forsteriVs$ and $\forsteriVc$, while independent current and voltage sources are denoted by $\forsteris$ and $\forstervs$ and their controlled counterparts are given by $\forsteric$ and $\forstervc$, respectively
\begin{subequations}
    \label{eq:forster_mna}
    \begin{alignat}{2}
        &\forsterA[C] \forsterdd{t} \forsterqC \big( \forsterAT[C] \forstervphi \big) &&\notag\\
        &\quad + \forsterA[R] \forstergR \big( \forsterAT[R] \forstervphi, \forsterAT[C V] \forstervphi, \forsteriL, t \big) &&\notag\\
        &\quad + \forsterA[L] \forsteriL + \forsterA[V_s] \forsteriVs + \forsterA[V_c] \forsteriVc &&\notag\\
        &\quad + \forsterA[I_s] \forsteris(t) + \forsterA[I_c] \forsteric \big( \forsterAT\! \forstervphi, \forsteriL, \forsteriVc, t \big) &&= \forstermb{0} \label{eq:forster_mna.1}\\
        &\forsterdd{t} \forsterphiL(\forsteriL) - \forsterAT[L] \forstervphi &&= \forstermb{0} \label{eq:forster_mna.2}\\
        &\forsterAT[V_s] \forstervphi - \forstervs(t) &&= \forstermb{0} \label{eq:forster_mna.3}\\
        &\forsterAT[V_c] \forstervphi - \forstervc \big( \forsterAT[C V_s] \forstervphi, \forsteriL, t \big) &&= \forstermb{0}. \label{eq:forster_mna.4}
    \end{alignat}
\end{subequations}

On an abstract level, MNA can be characterized as a differential-algebraic equation (DAE)~\cite{tischendorf1999}. Such DAEs differ from the more common ordinary differential equations (ODEs), in that their solutions underlie additional (hidden) constraints~\cite{hairer1996}. (Think e.g.\ of Kirchhoff's laws in the context of electric circuits.) As such, it becomes important to ensure that these constraints are also fulfilled by the learned solutions, to avoid physically invalid predictions.
While there already exists a family of so-called index-aware model order reduction methods which ensure that these constraints are upheld, see e.g.~\cite{ali2014}, they only work well for linear problems. Electric circuits often exhibit strongly nonlinear behavior however, as the example in~\cref{sec:forster_ex} will demonstrate.
In order to bridge this gap, previous works already investigated replacing classical model order reduction methods with machine learning approaches, see e.g.~\cite[Chapters 9 and 10]{benner2021} or~\cite{hesthaven2018}. We now show that in the case of electric circuits, one can obtain a non-intrusive version of the index-aware learning approach from~\cite{forster2023c}.

For completeness, we note that the idea of using neural networks for emulating the behavior of electric circuits goes back over twenty years~\cite{zaabab1995,meijer1996} and more recently, also Gaussian processes~\cite{rasmussen2006} have been employed to this end~\cite{wang2018,sanabria2020}. However all of these approaches typically treat the circuit simulator as a mere black box. In contrast, our approach leverages the mathematical structure of MNA to ensure that the learned solutions also fulfill the underlying constraints.

\section{Non-intrusive index-aware learning}
\label{sec:forster_ial}
The index-aware learning approach from~\cite{forster2023c} consists of three steps. First, the DAE of interest is decoupled into its differential and algebraic parts. Afterwards, only the (parameter dependent) trajectories of the differential variables corresponding to the differential part are learned. Finally, any required algebraic variables are reconstructed using the algebraic part in a third step. This last step is what ultimately ensures that the underlying constraints of the solution are also fulfilled by the predictions.

The key difficulty associated with this approach is efficiently computing the decoupling. The original description in \cite{forster2023c} is based on the dissection index \cite{jansen2014} and generally requires a detailed analysis of the DAE. Fortunately, for electric circuits modeled by~\cref{eq:forster_mna}, the procedure described in~\cite{forster2026a} provides such a decoupling in a topological, i.e.~graph-based and non-intrusive, manner. Some additional (topological) assumptions are required for this topological decoupling to exist, see~\cite[Assumptions 1, 2 and 3]{forster2026a}. These e.g.~ensure that the index of the circuit does not exceed two, compare also~\cite{estevez2000a}.

This alone does not yet suffice for a completely non-intrusive implementation however, as the reconstruction still requires computational access to the algebraic part. Thankfully, also this obstacle can be removed when dealing with electric circuits described by~\cref{eq:forster_mna}, by leveraging the results from~\cite{estevez2000b,garcia2022}. There, it was proved that two implicit Euler steps always suffice to turn an inconsistent initial condition into a consistent solution state for the equations of MNA as given in \cref{eq:forster_mna}, so long as the index does not exceed two, compare also the assumptions listed in~\cite[Theorem 4.1]{estevez2000a}. Instead of intrusively reconstructing the algebraic variables using the algebraic part, one can therefore simply perform one or two small implicit Euler steps depending on the index of the original system, while using zero initial conditions for the algebraic variables. One should note that this trades off computational effort for being non-intrusive however, as the implicit Euler steps need to performed with the full system, whereas individual algebraic variables may usually be reconstructed based on a small subset of the equations making up the algebraic part, due to the local nature of most of the equations making up~\cref{eq:forster_mna}.
Lastly, we note that the implicit Euler method is available as a standard solver option in all common SPICE-based circuit simulators, ensuring that this approach is truly non-intrusive in practice.

\section{Error analysis}
\label{sec:forster_ea}
In addition to ensuring consistent, i.e.~constraint-fulfilling, predictions, index-aware learning may also provide other benefits. Firstly, the (output) dimension of the learning problem is reduced, since only the differential variables have to be learned.
(In fact, also the number of parameters (input dimension) can decrease, as the differential variables may not depend on all parameters).
Moreover, the trajectories of the differential variables are often easier to learn than those of the algebraic variables, since they correspond to solutions of integration problems, rather than purely algebraic or even differentiation problems, which may arise for higher index components.
In order to illustrate these points and to provide a more rigorous error analysis, we now introduce the decoupling in formal terms.

The topological decoupling from~\cite{forster2026a} provides a linear variable transformation
\begin{align*}
    \forstermb{x} = \forstermb{T} \begin{bmatrix}
        \forstermb{x}_0\\
        \forstermb{x}_1\\
        \forstermb{x}_2
    \end{bmatrix}
\end{align*}
that splits the unknowns $\forstermb{x}^\forsterT = [\forstervphi^\forsterT, \forsteriL^\forsterT, \forsteriV^\forsterT]$ of~\cref{eq:forster_mna} into the differential (index zero) variables $\forstermb{x}_0$ and the algebraic (index one and two) components $\forstermb{x}_1$ and $\forstermb{x}_2$. Furthermore, the decoupling also provides a way to split~\cref{eq:forster_mna} itself into the following three systems
\begin{subequations}
    \label{eq:forster_dec}
    \begin{alignat}{2}
        \forsterdd{t} \forstermb{x}_0 &= \forstermb{f}_0(\forstermb{x}_0, \forstermb{x}_1, t, \forstermb{p}) \label{eq:forster_dec.1}\\
        \forstermb{0} &= \forstermb{f}_1(\forstermb{x}_0, \forstermb{x}_1, t, \forstermb{p}), \quad &\det \left( \forsterpp[\forstermb{f}_1]{\forstermb{x}_1} \right) &\neq 0 \label{eq:forster_dec.2}\\
        \forstermb{0} &= \forstermb{f}_2(\forstermb{x}_0', \forstermb{x}_1', \forstermb{x}_0, \forstermb{x}_1, \forstermb{x}_2, t, \forstermb{p}), \quad &\det \left( \forsterpp[\forstermb{f}_2]{\forstermb{x}_2} \right) &\neq 0, \label{eq:forster_dec.3}
    \end{alignat}
\end{subequations}
where we also explicitly included the parameters $\forstermb{p}$ and the $\cdot'$ denote time derivatives. We note that the conditions at the ends of~\cref{eq:forster_dec.2,eq:forster_dec.3} guarantee that the respective systems are locally solvable for the algebraic variables by the implicit function theorem.

As described in~\cref{sec:forster_ial}, our approach makes use of this decoupling by only learning the differential variables $\forstermb{x}_0$ and then reconstructing any required algebraic variables in an appropriate fashion. In this context, two different errors may play a role in evaluating the quality of the learned approximation and the performance of our approach. Firstly, we consider the consistency error
\begin{align*}
    e_\forstermr{c} \coloneqq \begin{Vmatrix}
        \forstermb{f}_1(\forstermb{x}_0, \forstermb{x}_1, t, \forstermb{p})\\
        \forstermb{f}_2(\forstermb{x}_0', \forstermb{x}_1', \forstermb{x}_0, \forstermb{x}_1, \forstermb{x}_2, t, \forstermb{p})
    \end{Vmatrix}
\end{align*}
that measures how accurately the inherent constraints of the DAE are fulfilled and secondly, we consider the approximation error
\begin{align*}
    e_\forstermr{a} \coloneqq \| \hat{x} - x \|
\end{align*}
between some component $x$ of the solution vector and its learned approximation $\hat{x}$. In both cases $\| \cdot \|$ stands in for some appropriate norm. For the examples in~\cref{sec:forster_ex}, we will work with the 2-norm based on a finite number of random samples.

While the consistency error is very directly tackled by our approach, we also provide some expectation for the reconstruction errors of the algebraic variables, assuming an approximation $\forstermbh{x}_0$ of the differential variables with
\begin{align*}
    \| \forstermbh{x}_0 - \forstermb{x}_0 \| \leq \varepsilon.
\end{align*}
Treating $\forstermb{x}_1$ as an implicit function of $\forstermb{x}_0$ and assuming $\forstermb{f}_1$ smooth enough in~\cref{eq:forster_dec.2}, the implicit function theorem implies
\begin{align*}
    \forsterpp[\forstermb{x}_1]{\forstermb{x}_0} = -\left( \forsterpp[\forstermb{f}_1]{\forstermb{x}_1} \right)^{-1} \forsterpp[\forstermb{f}_1]{\forstermb{x}_0}.
\end{align*}
Performing a Taylor expansion to first order about the true solution $\forstermb{x}_0$ then yields
\begin{align*}
    \forstermb{x}_1(\forstermbh{x}_0) \approx \forstermb{x}_1(\forstermb{x_0}) + \forsterpp[\forstermb{x}_1]{\forstermb{x}_0} (\forstermbh{x}_0 - \forstermb{x}_0).
\end{align*}
Assuming $\varepsilon$ small enough, this gives
\begin{align*}
    \| \forstermbh{x}_1 - \forstermb{x}_1 \| \leq L_{1,0} \| \forstermbh{x}_0 - \forstermb{x}_0 \| \leq L_{1,0} \varepsilon
\end{align*}
as an estimate for the approximation error of the index one algebraic variables, where $L_{1,0} \coloneqq \sup \forsterpp[\forstermb{x}_1]{\forstermb{x}_0}$. A similar argument leads to
\begin{align*}
    \| \forstermbh{x}_2 - \forstermb{x}_2 \| \leq L_{2,0'} \| \forstermbh{x}_0' - \forstermb{x}_0' \| + L_{2,1'} \| \forstermbh{x}_1' - \forstermb{x}_1' \| + (L_{2,0} + L_{2,1} L_{1,0}) \varepsilon
\end{align*}
as a bound for the approximation error of the index two variables, where the constants $L_{i,j}$ are defined analogously to $L_{1,0}$ before. Here, the situation is of course complicated by the derivative terms, which cannot be bounded based on information about the approximation error of the differential variables alone.

\section{Example}
\label{sec:forster_ex}
In the following, we illustrate the workflow of our approach and its theoretical properties using the filtered buck converter from~\cref{fig:forster_filtered_buck} as an example.
The corresponding decoupled system, compare also~\cref{eq:forster_dec}, is given by
\begin{align*}
    \forsterdd{t} \begin{bmatrix}
        \varphi_4\\
        i_\forstermr{L_f}\\
        i_\forstermr{L}
    \end{bmatrix} &= \begin{bmatrix}
        (i_\forstermr{L} - G \varphi_4) / C\\
        (\varphi_1 - \varphi_2) / L_\forstermr{f}\\
        (\varphi_3 - \varphi_4) / L
    \end{bmatrix}\\
    \forstermb{0} &= \begin{bmatrix}
        \varphi_1 - v_\forstermr{s}(t)\\
        G_\forstermr{f} \varphi_2 + g_\forstermr{S}(t) (\varphi_2 - \varphi_3) - i_\forstermr{L_f}\\
        -g_\forstermr{S}(t) (\varphi_2 - \varphi_3) + g_\forstermr{D}(-\varphi_3) \varphi_3 + i_\forstermr{L}
    \end{bmatrix}\\
    0 &= C_\forstermr{f} \forsterdd{t} \varphi_1 + i_\forstermr{L_f} + i_\forstermr{V_s},
\end{align*}
where the transistor is modeled as a time dependent resistor $g_\forstermr{S}(t)$ that interpolates between low and high resistance states based on a switching frequency of $500\,\forstermr{kHz}$ and a duty cycle of $12/5$, while the diode is similarly modeled as a nonlinear resistor interpolating between low and high resistance states based on the applied voltage. For further details, we refer to the code~\cite{forster2026g}. The left plot in~\cref{fig:forster_filtered_buck_simulation_cod} shows exemplary trajectories for some of the solution variables of the filtered buck converter. It also illustrates a common feature of electric circuit solutions, namely that the trajectories evolve on multiple time scales.

\begin{figure}[t]
    \centering
    \includegraphics[height=0.15\textheight]{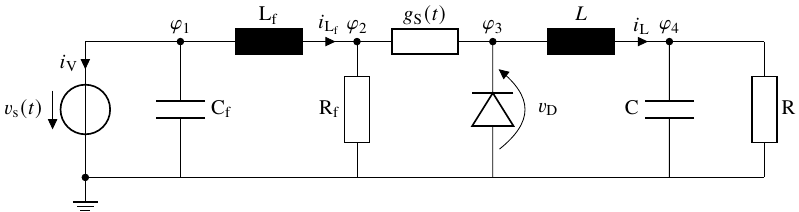}
    \caption{Filtered buck converter with the transistor (switch) modeled as a time dependent resistor $g_\forstermr{S}(t)$ and the diode modeled by a nonlinear resistor $g_\forstermr{D}(v_\forstermr{D})$.}
    \label{fig:forster_filtered_buck}
\end{figure}

For our example, we learn the solution of the filtered buck converter using Gaussian processes (GPs) and more specifically, we employ a simple sequential design strategy as e.g.~outlined in~\cite[Section 8.2]{rasmussen2006} or~\cite{burnaev2015,forster2025c}, for each solution component individually.
The right side of~\cref{fig:forster_filtered_buck_simulation_cod} illustrates the curse of dimensionality, which represents a further difficulty that arises when considering the dependence of the solution on additional parameters.
Here, we first learn the trajectory of $\varphi_4$ only over time (1D), then while also varying the input voltage $v_\forstermr{s}$ between $11$ and $13\,\forstermr{V}$ (2D) and finally when additionally varying the load resistance between $9$ and $11\,\forstermr{\Omega}$ (3D). One can clearly see that the rate of convergence deteriorates significantly when moving from 2D to 3D, which can be attributed to the curse of dimensionality. This simultaneously showcases that different parameters may impact the solution, and therefore the convergence, to a different extent.

\begin{figure}[b]
    \centering
    \includegraphics[height=0.35\textwidth]{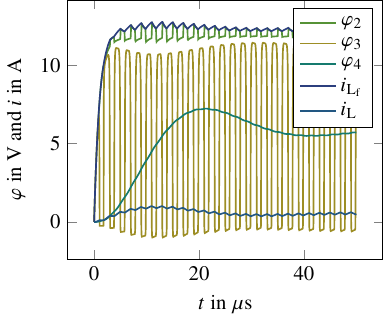} \hspace{3em} \includegraphics[height=0.35\textwidth]{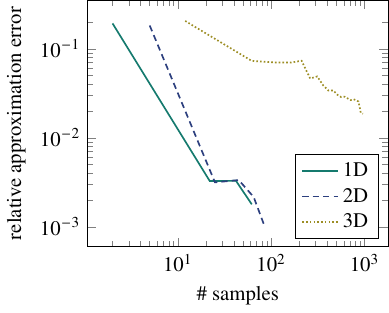}
    \caption{Exemplary solution trajectories for the filtered buck converter from~\cref{fig:forster_filtered_buck} (left) and illustration of the curse of dimensionality when learning $\varphi_4$ depending on different numbers of parameters (right).}
    \label{fig:forster_filtered_buck_simulation_cod}
\end{figure}

In the following, we stick with the 3D case and learn all solution variables, but only up to a relatively large tolerance of $5\,\%$. This ensures the computational effort stays manageable and essentially corresponds to only learning the long time scale (low frequency) content of the solution.
We can now investigate the behavior of the two error measures introduced in~\cref{sec:forster_ea}, when evaluating the GP for a random parameter combination that was not part of the training data.
First, we consider the consistency error, which can be seen on the left of~\cref{fig:forster_filtered_buck_consistency_approximation}. As expected, we see that the consistency error is large when learning all variables directly (learned), whereas it stays below or at the level of the consistency error of the simulation (simulated) when reconstructing based on the algebraic part (algebraic) or using the implicit Euler method (Euler). Here, we used a time step of $10\,\forstermr{ns}$ for the reference simulation and training data generation and a time step of $10\,\forstermr{ps}$ for the reconstruction. For the nonlinear solver, we used a relative tolerance of $10^{-6}$ and an absolute tolerance of $10^{-12}$.

\begin{figure}[t]
    \centering
    \includegraphics[height=0.35\textwidth]{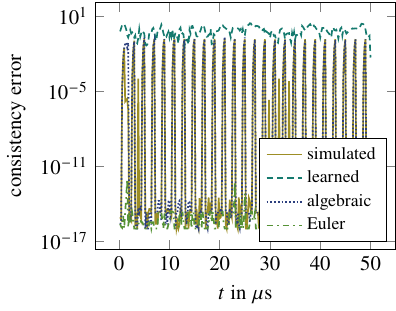} \hspace{3em} \includegraphics[height=0.35\textwidth]{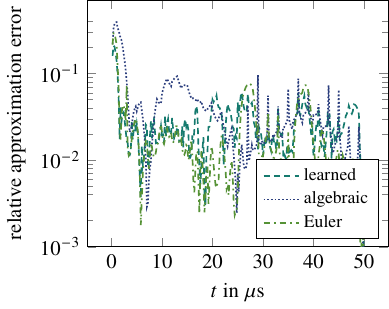}
    \caption{Consistency and approximation errors.}
    \label{fig:forster_filtered_buck_consistency_approximation}
\end{figure}

The corresponding relative approximation errors can be seen on the right of~\cref{fig:forster_filtered_buck_consistency_approximation}. The important thing to note is that the errors based on either reconstruction method behave similarly to those of the directly learned solutions, indicating that the reconstructions induce an error on the same order as that of the learning, in line with the argument from~\cref{sec:forster_ea}. We further note that the advantage of performing the reconstruction over directly learning the solution is twofold. Firstly, one only has to learn the differential variables $\varphi_4$, $i_\forstermr{L_f}$ and $i_\forstermr{L}$, while secondly, these variables also exhibit easier to learn behavior than e.g.~$\varphi_2$ and $\varphi_3$, as illustrated on the left of~\cref{fig:forster_filtered_buck_simulation_cod}.

Lastly, we collect some brief remarks about computational complexity. Depending on the number of samples required for achieving a relative approximation error of $5\,\%$, predicting a complete trajectory based on the GP takes between $0.1$ and $3.2\,\forstermr{s}$, while evaluating the reconstruction at a single point in time takes around $0.3\,\forstermr{ms}$. In comparison, integrating the circuit equations in time to obtain a similar relative error takes only $0.2\,\forstermr{s}$.
As a consequence, the GP even loses time compared to a conventional simulation in this case, however there are a few important caveats about this. First, since we use an individual GP for each solution variable, they can all be evaluated in parallel, such that even the (comparably expensive) prediction cost may win out over conventional methods for large enough circuits.
Furthermore, one may preprocess the training data to extract the low frequency components more explicitly, using e.g.~moving averages or Fourier coefficients, to speed up the predictions.
Finally, we emphasize that our approach can be combined with any approximation method of choice, such that GPs may be swapped with any other method that outperforms them in the given context.

\section*{Acknowledgment}
This work is supported by the Graduate School CE within the Centre for Computational Engineering at Technische Universit{\"a}t Darmstadt and the ECSEL Joint Undertaking (JU) under grant No.~101007319.

\bibliographystyle{spmpsci}
\bibliography{/home/peter/phd/biblio}

\clearpage
\section*{Appendix}
On top of the already mentioned benefits of constraint-fulfilling predictions and reducing the output dimension of the learning problem, index-aware learning may also decrease the input dimension of the learning problem, as mentioned in~\cref{sec:forster_ea}. This was previously indicated in~\cite{forster2023c}, but not rigorously analyzed there. Formally speaking, this decrease in the input dimension occurs when the differential part~\cref{eq:forster_dec.1} only depends on a subset of the parameters $\forstermb{p}$, such that
\begin{align*}
    \forstermb{f}_0(\forstermb{x}_0, \forstermb{x}_1, t, \forstermb{p}) = \forstermb{f}_0(\forstermb{x}_0, \forstermb{x}_1, t, \forstermb{p}_0),
\end{align*}
where the differential parameters $\forstermb{p}_0$ are a subset of all the parameters $\forstermb{p}$. This then implies that the remaining \emph{algebraic} parameters $\forstermb{p}_{1,2}$, need not be considered when learning the differential variables $\forstermb{x}_0$, as the latter only depend on $\forstermb{p}_0$.

In order to better understand when this occurs in the context of MNA, we now provide a topological characterization of when the values of particular circuit elements\footnote{This includes nonlinear elements and cases where the element model depends on external parameters.} become algebraic parameters $\forstermb{p}_{1,2}$. Analogous to the topological decoupling from~\cite{forster2026a}, the $\forstermb{p}_{1,2}$ may then be identified by only looking at the graph of the circuit. The characterization is split into the following three propositions, that respectively state which capacitors, inductors and resistors only contribute algebraic parameters.

\begin{proposition}[Algebraic capacitors]
    \label{prop:ac}
    Capacitors contribute only algebraic parameters if there exists a $\forstermr{V_s}$-only path connecting their terminals, i.e.~a path that consists of only independent voltage sources.
\end{proposition}

\begin{proposition}[Algebraic inductors]
    \label{prop:ai}
    Inductors contribute only algebraic parameters if they are part of an $\forstermr{L}$-$\forstermr{I_s}$ cutset, i.e.~a cutset that consists of only inductors and independent current sources.
\end{proposition}

\begin{proposition}[Algebraic resistors]
    \label{prop:ar}
    Similar to capacitors, resistors contribute only algebraic parameters if there exists a $\forstermr{V_s}$-only path connecting their terminals.
\end{proposition}

The filtered buck from~\cref{fig:forster_filtered_buck} also serves as an example of~\cref{prop:ac}, as the capacitor $\forstermr{C_f}$ only contributes algebraic parameters. The proofs all closely follow the ideas presented in~\cite{forster2026a}, hence we simply refer to the relevant parts of~\cite{forster2026a} to avoid reintroducing cumbersome notation.

\begin{proof}[Proof of~\cref{prop:ac}]
    Recalling the differential part corresponding to~\cref{eq:forster_dec.1} that was derived in~\cite[Eq. (18)]{forster2026a}, we observe that the value of a given capacitor does not appear in the differential part if the matrix product $\forsterPT[C] \forsterQT[V_s] \forsterA[C]$ contains a zero row at the index corresponding to that capacitor. Here, $\forsterA[C]$ is the incidence matrix describing capacitors as defined in~\cite[Definition 2]{forster2026a} and $\forsterPT[C]$ and $\forsterQT[V_s]$ are basis matrices as defined in~\cite[Definition 3]{forster2026a} and the proof of~\cite[Theorem 1]{forster2026a}.

    Using insights from the proof of~\cite[Proposition 3]{forster2026a}, one then sees that it suffices to consider zero rows of $\forsterQT[V_s] \forsterA[C]$, as multiplying by $\forsterPT[C]$ may not introduce additional ones. According to the same proof, such zero rows are generated by capacitors that possess a $\forstermr{V_s}$-only path connecting their terminals. As capacitors do not appear in the algebraic part corresponding to~\cref{eq:forster_dec.2}, given by~\cite[Eqs. (4), (9), (10), (12) and (14)]{forster2026a}, this represents the only condition for their values to become algebraic parameters.
\end{proof}

\begin{proof}[Proof of~\cref{prop:ai}]
    For inductors, we observe a similar argument, however this time we require zero rows in $\forsterW[L]$, compare~\cite[Eq. (18b)]{forster2026a}. Looking at the definition of $\forsterW[L]$ from~\cite[Eq. (24)]{forster2026a} together with the proofs of~\cite[Theorem 1 and Proposition 3]{forster2026a}, we see that zero rows in $\forsterW[L]$ correspond to inductors that are not part of any $\forstermr{L}$-$\forstermr{I_s}$ loop in the graph that is created by contracting all circuit elements that are not inductors or independent current sources. This in turn implies that such zero rows correspond to $\forstermr{L}$-$\forstermr{I_s}$ cutsets in the original graph, as the inductors and independent current sources in question would have otherwise formed a loop with elements of different types in the original graph and therefore they would have been contracted away alongside these.

    As inductors also do not appear in the algebraic part, this again represents the only condition for their values to become algebraic parameters. Finally, we remark that for multiport or coupled inductors, such as those modeling transformers e.g., all of the relevant inductors need to be part of $\forstermr{L}$-$\forstermr{I_s}$ cutsets for their values to become algebraic parameters in general.
\end{proof}

\begin{proof}[Proof of~\cref{prop:ar}]
    Lastly, we note that the value of a given resistor does not appear in the differential part if that resistor corresponds to a zero row\footnote{In principle, zero rows in $\forsterPT[C] \forsterQT[V_s] \forsterA[R]$ suffice, however the extra zero rows induced by $\forsterPT[C]$ will not in general appear in the algebraic part. Thus, only zero rows of $\forsterQT[V_s] \forsterA[R]$ are relevant in practice.} in $\forsterQT[V_s] \forsterA[R]$, compare again~\cite[Eq. (18)]{forster2026a}. In fact, such zero rows also suffice for the same resistor to not appear in the algebraic part, as resistors only appear in products of the form $\bullet \forsterQT[V_s] \forsterA[R]$ in~\cite[Eqs. (10) and (12)]{forster2026a}, where $\bullet$ denotes a matrix of matching dimensions.
\end{proof}
\end{document}